\documentclass[reqno,10pt]{amsart}

\usepackage{amsfonts}
\usepackage{amscd,amsmath,amsthm,amsfonts,amsmath}
\usepackage{hyperref}
\usepackage{srcltx}

\usepackage{color}
\usepackage[svgnames]{xcolor}

\newtheorem{theorem}{Theorem}
\theoremstyle{plain}

\newtheorem{corollary}[theorem]{Corollary}
\newtheorem{definition}[theorem]{Definition}

\newtheorem{lemma}[theorem]{Lemma}
\newtheorem{proposition}[theorem]{Proposition}

\numberwithin{equation}{section}

\def\evt{}

\newcommand{\modu}{\,\,\,\mathrm{mod}\,}

\def\sisalter{[f,g]}
\def\sq{(x_n)}
\def\evt{}
\def\rojo{}

\def\DD{\Delta}

\def\kk{\kappa}
\def\Zpi{\Z_{p_i^{s_i}}}

\def\kkdef{\beta (d_1-1) - \alpha (d_2 -1)}

\def\sisalter{[f,g]}
\def\sq{(x_n)}
\def\Per{\mathrm{Per}}

\def\N{\mathbb{N}}
\def\R{\mathbb{R}}
\def\Z{\mathbb{Z}}
\def\Q{\mathbb{Q}}
\def\menos{\backslash}
\def\O{\mathbb{O}}

\renewcommand{\S}{\mathbb{S}}

\def\deg{\textrm{deg}}

\def\gcd{\mathrm{gcd}}

\begin{document}

\title[Computing odd periods of alternating systems of affine circle maps]{Computing odd periods of alternating systems of affine circle maps}

\author{J. S. C\'anovas}
\address{Departamento de Matem\'atica Aplicada y Estad\'{\i}stica\\
Universidad Polit\'ecnica de Cartagena,  Campus Muralla del Mar,
30203--Cartagena (Spain).} \email{Jose.Canovas@upct.es}
\author{A. Linero Bas}
\address{Universidad de Murcia, Campus de Espinardo, 30100-Murcia (Spain).}
\email{lineroba@um.es}
\author{G. Soler L\'opez}
\address{Departamento de Matem\'atica Aplicada y Estad\'{\i}stica\\
Universidad Polit\'ecnica de Cartagena,, Paseo Alfonso XIII, 52,
30203--Cartagena (Spain).} \email{Gabriel.Soler@upct.es }

\maketitle

\begin{abstract} Let $f,g$ be affine  circle maps and let  $[f,g]$ be the alternating system
generated by $f$ and $g$. We present an algorithm to compute the  periodic structure of $[f,g]$. This study complements the papers
\cite{calisoalternated1,calisoalternated2}.

%


{\bfseries
February 2, 2019. 

The final version of this paper  will be published in Journal of Difference Equations and Applications.}
\end{abstract}


\section{Introduction}

\def\S{\mathbb{S}} 

Given a topological space $X$, let $C(X)$ be the set of continuous maps defined from  $X$ to $X$, and let  
$(f_n)$ be  a sequence of
maps from $C(X)$. A \emph{non--autonomous  dynamical system} is built
by iterating  sequentially the maps  $f_n$, that is, the \emph{orbits} are:
\begin{equation}\label{Eq:NonAut}
x_{n+1}=f_n(x_n), n=0,1,2,\ldots,
\end{equation}
with $x_0\in X$. When the sequence $(f_n)$ is periodic  of period $p$ then the system,
denoted by $[f_0,f_1,\dots,f_{p-1}]$,  is called an \emph{alternating system.}  This type of
dynamical systems as well as the general non-autonomous case have received special attention in the literature due to
their potential applications to
natural and social sciences (for instance, see \cite{CuHen,ElaSa,CaCoSa,Krau}). To this respect, take into account
that the deterministic law can be dependent on the time due, for instance,
to effects of seasonality producing changes in the environment (although we are interested in
discrete systems, for non-autonomous continuous models dealing with the effect of seasonality, the reader
is referred to the survey \cite{Buono}), or due to the appearance of new information or technical aspects
in economical or social features.
Moreover, from the theoretical point of view, many authors have studied the qualitative properties of the orbits generated by these systems,
see for example \cite{kolsnohatro,alsharawiEtAl,canovaslinero06,CanoNonAut,daniellosteele,joao,calisoalternated1,calisoalternated2}.

We focus here on alternating systems generated by 2-periodic sequences of maps on the circle
$\S^1:=\{z\in\mathbb{C}:|z|=1\}$. Namely, $X= \S^1$,  $f,g:\S^1\to \S^1$ are continuous and the orbits generated by $[f,g]$ are sequences $(x_n)$ given by:
$$
x_{n+1}=
\left\{\begin{array}{l}
f(x_n) \textrm{ if $n$ is even,}\\
g(x_n) \textrm{ if $n$ is odd.}
\end{array}
\right.
$$
From the different qualitative aspects of alternating systems, we are interested in the topic
of the periodicity and in the search of periodic structures associated to these systems.
Recall that the sequence
$(x_n)$ and  $x_0$ are \emph{periodic} if there exists $m\in \N:=\{1,2,\ldots\}$ such that $x_{n+m}=x_n$
for all $n=0,1,2,\ldots,$ and the smallest positive
integer $m$ satisfying the previous condition is called the \emph{period}  or \emph{order} of $\sq$ and $x_0$.
In addition, we say that $x_0$, or the sequence $(x_n)$, is \emph{$m$--periodic}.
By $\Per\sisalter$ we understand the set of periods of $\sisalter$, that is,
the set of positive integers $m$ for which there exists $x_0\in X$ such that it is $m$--periodic. Notice that 
$\Per\sisalter=\Per[g,f]$. 
As usual we will use the following notation for the iterates of $f$: $f^1=f, f^{n} = f\circ f^{n-1}, n\geq 2$. By $f^0$ we mean the identity
on $X$. Observe that when $f=g$, we receive a classical \emph{autonomous dynamical system} and the analogous definitions on periodicity 
(for instance, see \cite{alseda}), in this case we will
use the standard notation and $[f,f]$ will be replaced simply by $f$.

In order to manage circle maps, it is necessary to introduce the notion of lifting
(the reader is referred to \cite{alseda} for a deeper development of this topic).
Consider the standard universal
covering $e:\mathbb{R}\rightarrow\mathbb{S}^1$ given by
$e(x)=e^{2\pi ix}.$ If $f\in C(\mathbb{S}^1),$ we find a (non unique) map
$F:[0,1]\rightarrow \mathbb{R}$ such that the diagram
\[
\begin{array}{ccc}
\lbrack 0,1] & \overset{F}{\longrightarrow } & \mathbb{R} \\ e\downarrow &
& \downarrow e \\
\mathbb{S}^{1} & \overset{f}{\longrightarrow } & \mathbb{S}^{1}%
\end{array}%
\]
commutes.  Then, $F$ is said to be a  \emph{lifting} of $f.$ Realize that
$e(0)=e(1)=1,$ so $e(F(1))=f(e(1))=f(e(0))=e(F(0))$. Thus, $d:=F(1)-F(0)\in\mathbb{Z}.$
The integer $d$ is said to be the \emph{degree}
of $f,$ and we denote it by $\deg(f).$ Moreover, we can extend
the lifting $F$ from $[0,1]$ to $\mathbb{R}$ by considering
$\widetilde{F}:\mathbb{R}\rightarrow\mathbb{R}$ as $\widetilde{F}(x)=F(x-[x])+[x]\deg(f),$
where $[\cdot]$ is the integer part of a real number $x$. To simplify the notation, in the sequel 
we will identify $\widetilde{F}$ with $F$.
It is well known that   $\deg(f\circ g)=\deg(f)\deg(g)$ and then  $\deg(f^m)=(\deg(f))^m$ for all
$m\geq 1$.

In this work we deal with \emph{affine circle maps}. A map $f\in C(\S^1)$ is said to be \emph{affine} if its  lifting is of the type
$F(x)=dx+\alpha$ for some $d\in\mathbb{Z}$ and $\alpha \in \mathbb{R}$. Hence
\begin{equation}\label{Eq:f}
f(z)= z^d e^{2\pi\alpha i}, z\in\mathbb{S}^1.
\end{equation}
The periodic structure of these maps 
is  given in \cite[Lemma 2]{calisoalternated1}:

\begin{lemma}\label{periodosrotaciones}
 Let $f\in C(\S^1)$ be  an affine map with associate lifting $F(x)=d x+\alpha$. Then:
\begin{enumerate}
 \item If $|d|\geq 3$ or $d=2$, then $\mathrm{Per} (f)=\N$.
 \item If $d=-2$, then $\mathrm{Per} (f)=\N\menos \{2\}$.
 \item If $d=-1$, then $\mathrm{Per} (f)=\{1,2\}$.
 \item If $d=1$, and  $\alpha=\frac{M}N \in\Q$, with $\gcd(M,N)=1$ and $N>0$,
then $\mathrm{Per} (f)=\{N\}$. Otherwise, $\mathrm{Per} (f)=\emptyset.$
 \item If $d=0$, then $\mathrm{Per}(f)=\{1\}$.
\end{enumerate}
\end{lemma}

Even in the easiest case of alternating systems given by two affine circle maps on the circle,
 the characterization of the set of periods is hard. The present
authors have solved this problem in \cite{calisoalternated1,calisoalternated2}. We sketch the results
in Section~\ref{S:previousresults}. In Section~\ref{S:preliminar} we present some preliminary technical results and we give a
new determination for the odd periods of the system $[f,g]$, whose appearance is only possible whenever $f$ and $g$ commute. This new approach to the calculation of odd periods (based on the iterates of a suitable set $C$ of real numbers and on appropriate properties of congruences of integer numbers) is different from that developed in \cite{calisoalternated1,calisoalternated2} (in terms of a certain discrete map $\phi:\mathbb Z_{\Delta}\rightarrow \mathbb Z_{\Delta}$ in these last cases). Finally, our new approach originates a
direct algorithm, described in Section~\ref{S:algorithm}, to compute odd periods.

\section{Results on $\Per\sisalter$}\label{S:previousresults}

In what follows, $f$ and $g$ are affine   maps on $\S^1$ with respective liftings $F(x)=d_1x+\alpha$ and
$G(x)=d_2x+\beta$. We adopt the notation:
\begin{equation}
 \kk:=\kkdef\label{E:kappa}.
\end{equation}
Next, we introduce the characterization of the set of periods $\Per\sisalter$ through se\-veral results. They
are given depending on whether the degrees of $f$ and $g$ coincide (see Theorem~\ref{MP:samedegree})
or not. In the last case, there are subcases depending on the value $\kk$
(see Theorems~\ref{P:main1} and \ref{T:main}).

\begin{theorem}\label{MP:samedegree}(\cite[Theorem A]{calisoalternated1})
Let $f,g\in C(\S^1)$ be affine maps with liftings $F(x)=dx+\alpha$ and $G(x)=dx+\beta$, respectively,
and such that $\beta-\alpha\notin\mathbb{Z}$. Then:
\begin{itemize}
\item[(1)] If $|d|\geq 2$, then $\mathrm{Per} \sisalter=2\N$ (the set of even positive integers).
\item[(2)] If $|d|=1$, we distinguish the cases
\begin{enumerate}
 \item[(i)] If $\alpha+d\beta	\not\in\Q$  then $\mathrm{Per}\sisalter=\emptyset$.
 \item[(ii)] If $\alpha+d\beta=\frac{M}{N}\in\Q$ with $\textrm{gcd}(M,N)=1$ then $\mathrm{Per}\sisalter=\{2N\}$.
\end{enumerate}
\item[(3)] If $d=0,$ then $\mathrm{Per}\sisalter=\{2\}$.
\end{itemize}
\end{theorem}

When $f$ and $g$ have different degrees we must distinguish two cases depending on whether $\kk$, cf. Eq.~(\ref{E:kappa}),
is or is not an integer number. By $\mathbb{O}$ we denote the set of odd natural numbers.

\begin{theorem}\label{P:main1}(\cite[Theorem B]{calisoalternated1})
Let $f,g\in C(\S^1)$ be affine maps with liftings $F(x)=d_1 x +\alpha$ and $G(x)=d_2 x +\beta$, with $d_1\neq d_2.$
Assume that $\kappa \notin\mathbb{Z}$ and let $\delta:=d_1d_2$. Then:
\begin{enumerate}
\item If $|\delta|\geq 3$ or $\delta=2$, then $\mathrm{Per}[f,g]=2\mathbb{N}$.
\item If $\delta=-2$, then $\mathrm{Per}[f,g]=2\mathbb{N}\setminus\{4\}$.
\item If $\delta=0$, then $\mathrm{Per}[f,g]=\{2\}$.
\item If $\delta=-1$, then $\mathrm{Per}[f,g]=\{2,4\}$.
\end{enumerate}
\end{theorem}

\evt
\begin{theorem}\label{T:main}(\cite[Theorem C]{calisoalternated1})
Let $f,g\in C(\S^1)$ be affine maps with liftings $F(x)=d_1 x + \alpha$ and $G(x)=d_2 x +\beta,$ being $d_1\not=d_2$, $d_2\not=1$ and
$\kappa\in\Z$. Let
$$
A= \left\{
\begin{array}{l}
 \{4\} \textrm{ if $d_1d_2=-2$, }\\
\emptyset \textrm{ otherwise.}
\end{array}
\right.
$$
Let $\Lambda=\mathrm{Per}([f,g])\cap \mathbb{O}$. Then:
\begin{enumerate}
\item  If $d_1\in\{-1, 0, 1\}$, then $\mathrm{Per}([f,g])$ is described by the following table:

\begin{tabular}{|c||c|c|c|c|c|}
\hline
&\rojo $ d_1$&\rojo $d_2$&\rojo $\Lambda$&\rojo $\mathrm{Per}[f,g]$\\ \hline
(a)&0&&&\{1\}\\ \hline
(b)&1&$\notin\{-1,0,1\}$& $\{1\}$&$[2(\N\menos\{1\})\menos A]\cup\{1\}$ \\ \hline
(c)&1&-1&$\{1\}$& $\{1,4\}$\\ \hline
(d)&1&$\notin\{0,1\}$&$\{N\}$& $[2\N\menos A]\cup\{N\}$\\ \hline
(e)&1&-1&$\emptyset$& $\{2,4\}$\\ \hline
(f)&1&$\notin\{-1,0,1\}$&$\emptyset$&$2\N\menos A$\\ \hline
(g)&-1&$\in\{-3,-2\}$&$\{1\}$&$2(\N\menos\{1\})\cup\{1\}$\\ \hline
(h)&-1 & $\notin\{-3,-2,-1,0,1\}$ &$\{1\}$& $[2\N\menos A]\cup\{1\}$\\ \hline
(i)&-1&$\notin\{-1,0,1\}$&$\emptyset$&$2\N\menos A$\\ \hline
\end{tabular}

Here $N$ denotes an odd natural number bigger than $1$. In addition, the role of $d_1$ and $d_2$ in the above table can be changed.

\item \label{T:mainpartgeneral} If $\{d_1,d_2\}\cap\{-1,0,1\}=\emptyset$ then
 $$\mathrm{Per}[f,g]=[2\mathbb{N}\menos A] \cup \Lambda.$$
\end{enumerate}
\end{theorem}

These results completely characterize the set of even periods in $\Per\sisalter$.
By means of the study of a dynamical system defined  on $\Z_{\Delta}$ for $\Delta=|d_1-d_2|$
it is possible to say  what odd numbers are in $\Per(\sisalter)$.

\subsection{The set $\Lambda$ of odd periods in $\Per\sisalter$}\label{SS:odd}

Let $d_1,d_2$ be integers, $\alpha,\beta$ be real numbers, $\kk=\beta (d_1-1) - \alpha (d_2 -1)$,   $\Delta=|d_1-d_2|$  and define
$\phi_{d_i,\kappa}:\mathbb{Z}_{\Delta }\rightarrow \mathbb{Z}%
_{\Delta }$,  $i\in\{1,2\}$, by
\begin{equation}
 \label{Eq:phi}
 \phi_{d_i,\kappa}(m):=(d_{i}m+\kappa )\,\modu (\Delta ).
\end{equation}
Since $d_{1}\equiv d_{2}\modu (\Delta )$, we have $\phi_{d_1,\kappa}=\phi_{d_2,\kappa}=:\phi$ and
the following result connects $\Lambda=\mathrm{Per}([f,g])\cap \mathbb{O}$ and $\Per(\phi)$.

\begin{theorem}\label{T:oddPeriods}(\cite[Prop. 15, Th. 17]{calisoalternated1})
Let $f,g\in C(\mathbb{S}^{1})$ be affine maps with associate liftings $F(x)=d_{1}x+\alpha $ and $%
G(x)=d_{2}x+\beta $ and  $d_{1}\neq d_{2}$. If $\kappa\not\in\Z$ then
$\Lambda=\emptyset$, otherwise  $\kappa \in
\mathbb{Z}$ and   {\rojo $\Lambda=\mathrm{Per}(\phi )\cap \O$.}
\end{theorem}

Then, in order to describe the set $\Lambda=\O\cap\Per\sisalter$ it suffices to study the set $\Per(\phi)$.
Let $p,s\in \N$, and define

$$\sigma(p,s):=\left\{\begin{array}{lcl}
                      1,&\quad & \textrm{ if $p$ is odd or $p=2$ and $s=2$},\\
                      2,&\quad & \textrm{ otherwise}.
                     \end{array}
\right.$$

We introduce the first case, $\Delta=p^s$ with $p$ prime.
\evt
\begin{theorem}\label{main0}(\cite[Theorem A]{calisoalternated2})
 Let $\Delta=p^s$ where $p$ is a prime and $s\geq 1$  and let $\phi_{d,\kappa}:\Z_\Delta\to\Z_\Delta$ be
 defined by $\phi_{d,\kappa}(x)=dx+\kappa$,  $d,\kappa\in\Z_\Delta$. Then $\Per(\phi_{d,\kk})$ is one of the
 following sets:
\begin{enumerate}
 \item
 $\{1\}\cup \{Np^j\}_{j=0}^\alpha$ where $N$ is a divisor of $p-1$ and $\alpha\in\{0,1,\dots,s-\sigma(p,s)\}$.

 \item $\{p^\alpha\}$ for some $\alpha\in\{0,1,\dots,s\}$.

\end{enumerate}

Conversely, let $p$ be a prime, let $\Delta=p^s$ be with $s\geq 1$, and let $\mathcal{P}$ be one
of the above sets $\{1\}\cup \{Np^j\}_{j=0}^\alpha$ or $\{p^\alpha\}$, then there exists
$\phi_{d,\kk}:\Z_\Delta\to\Z_\Delta$ such that $\Per(\phi_{d,\kk})=\mathcal{P}.$ 

\end{theorem}

Next result helps us to determine exactly $\Per(\phi_{d,k})$ for fixed $d,\kk\in\Z_\Delta.$
\evt

\begin{theorem}\label{main1}(\cite[Prop. 3.1, Th. C]{calisoalternated2})
 Let $\Delta$ be a positive integer, $d,\kappa\in\Z_\Delta$, let $\phi_{d,\kappa}:\Z_\Delta\to\Z_\Delta$ be
 defined by $\phi_{d,\kappa}(x)=dx+\kappa$. Then we distinguish the following cases:

\begin{enumerate}

\item  For any  $\Delta\in\N$ we have  $\Per (\phi_{0,\kappa})=\{1\}$
 and  $\Per (\phi_{1,\kappa})=\{\frac\Delta{\gcd(\Delta,\kappa)}\}$.

\item If $\Delta\geq 3$ is even then $\Per (\phi_{\Delta-1,\kappa})=\{1,2\}$ if $\kappa$ is even and
$\Per (\phi_{\Delta-1,\kappa})=\{2\}$ if $\kappa$ is odd.

\item If $\Delta\geq 3$ is odd then $\Per (\phi_{\Delta-1,\kappa})=\{1,2\}$ .

\item When $\Delta=p^s$, $p$ prime, we have:

\hspace*{-1cm}\begin{minipage}{.3\textwidth}
{\small
\begin{center}
  \begin{tabular}{|c|l|c|}
  \hline
	\multicolumn{2}{|c|}{\rojo Conditions on $d,\Delta,\kappa$}	&\rojo $\Per(\phi_{d,\kappa})$		\\
  \hline
  $\gcd(d,\Delta)=1$&\begin{tabular}{l}
		      $\gcd(d-1,\Delta)=1$\\
		      $d^N\equiv 1\modu(p^\alpha),\,\alpha\geq 1$\\
		      $d^N\not\equiv1\modu(p^{\alpha+1})$\\
		      $N$ is the order of $d$ modulo $p$
		      \end{tabular}		&              $\{1\}\cup N\cdot\{p^j\}_{j=0}^{\max\{0,s-\alpha\}}$\\ \cline{2-3}
  &\begin{tabular}{l}
		      $\gcd(d-1,\Delta)>1$\\
		      $d\equiv1\modu(p^\alpha)$,\,\, $d\not\equiv1\modu(p^{\alpha+1})$\\
		      $\kappa\equiv0\modu(p^\beta)$,\,\, $\kappa\not\equiv0\modu(p^{\beta+1})$\\
		      $1\leq \alpha<s,\;0\leq \beta<s,$\\
		      \textrm{If $p=2$ this only works when $\alpha>1$}
		      \end{tabular}		&
						$\{p^{j}\}_{j=0}^{s-\alpha}$ if $\beta\geq \alpha$\\ \cline{3-3}
 &						&            $\{p^{s-\beta}\}$ if $\beta<\alpha$\\
 \hline
  $\gcd(d,\Delta)>1$&&$\{1\}$\\
  \hline
  \end{tabular}
\end{center}
}
\end{minipage}

\item For $\Delta=2^s\geq 3$, the missing cases  are:

\hspace*{-1cm}{\small
\begin{center}
  \begin{tabular}{|c|l|c|}
  \hline
	\multicolumn{2}{|c|}{\rojo Conditions on $d,\Delta,\kappa$}	&\rojo $\Per(\phi_{d,\kappa})$		\\
  \hline
\begin{tabular}{l}

 $d\equiv 1\modu(2),\,\,\,
d\not\equiv1\modu(2^2)$\\
 $\kappa\equiv 0\modu(2^\beta),\,\,\,\kappa\not\equiv 0\modu(2^{\beta+1})$\\
 $d^2\equiv1\modu(2^\gamma),\,\,\,d^2\not\equiv1\modu(2^{\gamma+1})$\\
		$0\leq\beta<s,\;\gamma\geq 3$\\

\end{tabular}
  &\begin{tabular}{l}
				$\beta=0$
		      \end{tabular}		&
						$\{2\}$ if $s\leq \gamma-1$\\ \cline{3-3}
 &						&$\{2^{s-\gamma+2}\}$ if
 $s> \gamma-1$\\
\cline{2-3}
  &\begin{tabular}{l}
		$\beta\geq 1$
		      \end{tabular}		&
		       $\{2^j\}_{j=0}^{\max\{1,s-\gamma+1\}}$ \\
 \hline
  \end{tabular}
\end{center}
}

\end{enumerate}
\end{theorem}

Finally, in the general case  we have a prime decomposition $\Delta =p_{1}^{s_{1}}p_{2}^{s_{2}}\dots p_k^{s_k}$. Let
$\phi_{d,\kk}:\Z_\Delta\to\Z_\DD$ be defined by $\phi_{d,\kk}(x)=dx+\kk$ and  take
$
d_i\equiv d\modu(p_i^{s_i}),\;
\kk_i\equiv \kk\modu(p_i^{s_i}),\;
\phi_{d_i,\kk_i}:\Zpi\to\Zpi,$ $i=1,2,\ldots,k.$

Using the Chinese Remainder Theorem, it can be proved:

\evt
\begin{theorem}\label{main4}
Let $\Delta =p_{1}^{s_{1}}p_{2}^{s_{2}}...p_{k}^{s_{k}}$ be a decomposition into prime factors.
Then, $n\in \mathrm{Per}(\phi _{d,\kappa })
$ if and only if $n=\mathrm{lcm}(n_{1},n_{2},...,n_{k})$ for $n_{i}\in \mathrm{%
Per}(\phi _{d,\kappa _{i}})$.
\end{theorem}
This result allows us to obtain a precise description of the set of periods in the general case
$\Delta =p_{1}^{s_{1}}p_{2}^{s_{2}}\dots p_k^{s_k}$. Since the description is something cumbersome,
the reader is referred to \cite[Theorem D]{calisoalternated2} to see all the details of such a general
description of $\mathrm{Per}[f,g]$.

\section{Technical results. Determination of odd periods}\label{S:preliminar}

To study the question of determining the set of odd periods of $[f,g]$ we need  to link periodic points of the
circle maps $f, g$ with ``lifted cycles" of their corresponding liftings $F, G$. In this direction,
we write $[F,G]$ to denote the set of sequences $(x_n)\subset \mathbb{R}$ obtained by applying $F$ and $G$ in an
alternated form.
A point $z\in\S^1$ is said to be a \emph{shared periodic point of $f$ and $g$ of order or period $m$} if and only if it is an $m$-periodic point of $f$ and $g$ whose orbit is shared by both maps,
that is, $f^i(z)=g^i(z)$ for any integer $i\geq 0$. Note that $f^m(z)=g^m(z)=z$ and $f^i(z)=g^i(z)\not=z$ for any $i\in\{1,2,\dots, m-1\}$.
We recall some interesting results we will need later for the determination of odd periods.

\begin{lemma}\label{L:ciclosf_ciclosF}
Let $f\in C(\mathbb{S}^1)$ and let $F:\mathbb{R}\rightarrow\mathbb{R}$ be an associate lifting of $f$. Then
$z\in\mathbb{S}$ is a periodic point of $f$ of order $p\geq 1$ if and only if
the orbit of $x$, with $e(x)=z$, under $F$ is a lifted cycle of period $p$, that is, $F^{p}(x)-x\in\mathbb{Z}$ and
$F^{s}(x)-x\notin\mathbb{Z}$ for all $0<s<p$.
\end{lemma}
\begin{proof}
 See \cite[Sections 3.2 and 3.3]{alseda}.
\end{proof}

By applying the  definition of shared periodic point and the previous lemma we obtain immediately:

\begin{proposition}\label{puntocompartido}
 Let $f,g\in C(\S^1)$ be with liftings $F$ and $G$, respectively. Then $z\in\S^1$ is a shared  periodic point of $f$ and $g$ of period $m$ if and only if for any $w\in e^{-1}(z)$ it holds $F^j(w)-G^j(w)\in\Z$ for all $j\geq1$, $F^m(w)-w\in \Z$ and $F^r(w)-w\not\in\Z$ if $0<r<m.$
\end{proposition}

In order to investigate what points from the circle generate periodic orbits of odd period, we need the following lemma stated 
in \cite[Lemma 6-(6)]{calisoalternated1}:

\begin{lemma}\label{L:Char_lift}\label{lemlinerosalvalifting}
Let $f, g\in C(\mathbb{S}^1)$ be with associate liftings $F,G:\mathbb{R}\rightarrow\mathbb{R}$, respectively. Let $(z_n)\subset \mathbb{S}^1$ be an orbit in the alternating system $[f,g]$,  let $x_0\in\mathbb{R}$ be such that $e(x_0)=z_0$ and let  $m$ be an odd number. Then $(z_n)\subset \mathbb{S}^1$ is a periodic sequence of period $m$ if and only if the following conditions hold:
\begin{enumerate}
 \item
$
(G\circ F)^{m}(x_0)-x_0\in\mathbb{Z},\, (G\circ F)^{s}(x_0)-x_0\notin\mathbb{Z}, \, 0< s <m.
$
\item
$
F^{m}(x_0)-x_0\in\mathbb{Z},\, F^{s}(x_0)-x_0\notin\mathbb{Z},\, 0< s <m.$
\item $G^{m}(x_0)-x_0\in\mathbb{Z},\, G^{s}(x_0)-x_0\notin\mathbb{Z},\, 0< s <m.$
\item
$
F^{i}(x_0)-G^{i}(x_0)\in \mathbb{Z}, \text{ for all } i\geq 1.
$
\end{enumerate}
\end{lemma}

From now on, our continuous circle maps $f$ and $g$ have liftings
$$F(x)=d_1x+\alpha\qquad\textrm{ and } \qquad G(x)=d_2x+\beta,$$ with $d_1\neq d_2$ and we will
be concerned only with $\kk\in\Z$, see Eq~(\ref{E:kappa}),
since it implies that $f$ and $g$ commutes, a necessary condition to obtain odd periods of
$\sisalter$. For more details, the reader is referred to \cite[Section~5]{calisoalternated1}.

In next result  we sketch the key conditions to obtain points generating periodic orbits of odd period,
see \cite[Lemma 16]{calisoalternated1}.

\begin{lemma}[The set $C$]  \label{L:Candi}
Let
\begin{equation}
C:=\left\{\frac{m+\beta-\alpha}{d_1-d_2}: m\in\mathbb{Z}\right\}. \label{conjuntoC}
\end{equation}
Then the  following properties hold:
\begin{enumerate}
 \item $C+\Z=C$.
 \item $F(C)\subseteq C$ and $G(C)\subseteq C$.
 \item If $z\in \S^1$ is a periodic point shared by $f$ and $g$, with odd period, and $z=e(x)$, then
$x=\frac{m+\beta -\alpha}{d_1-d_2}\in C,$ for some $m\in \mathbb Z.$
\item $F^s(x)-G^s(x)\in\mathbb{Z}$ for all $x\in C$ and for all integer $s\geq 1.$
\end{enumerate}
\end{lemma}

It is immediate to realize that

\begin{equation}
e(C):=\{e(x):x\in C\}=\left\{e^{2\pi i x}:x=\frac{m+\beta-\alpha}{d_1-d_2},
0\leq  m< |d_1-d_2|\right\}\label{conjuntoeC}
\end{equation}
is finite, and, by Lemma~\ref{L:Candi},  if $z_0\in e(C)$ and $(z_n)$ is the orbit
generated by $\sisalter$ then $z_n\in e(C).$ 

\begin{lemma}\label{L:RCandi} Under the above conditions, if $z_0\in e(C)$, then
its corresponding orbit $(z_n)$ in $\sisalter$ is either periodic or eventually periodic, that is, there exist
 $p\in\N$ and $n_0\geq 0$ such that $z_{n+p}=z_n$ for all $n\geq n_0$.
Additionally, $z_{n_0}$ is a shared periodic point of $f$ and $g$ of order or period $p.$
\end{lemma}
 \begin{proof}
The (eventual) periodicity of $z_0$ is guaranteed by the finiteness of $e(C)$ and Lemma~\ref{L:Candi}. Concerning the fact
of being a shared periodic point of $f$ and $g$, the last part of Lemma~\ref{L:Candi} implies that $f^s(z)=g^s(z)$ for all $s\geq 1$.
\end{proof}

Next result establishes an easy way to determine the period of a periodic point in $C$ that will be useful
in the performance of the algorithm. Remember that, given an initial condition $x_0$, by $(x_j)_j$ we denote
the lifted orbit under $[F,G]$ generated by $x_0$ ($x_1=F(x_0), x_2=G(x_1), \ldots$), and that we say that
$x_0$ is a periodic point in $[F,G]$ whenever $x_{n+q}\equiv x_n$ for all $n\geq 0$ and some $q\in\N$.

\begin{lemma}\label{L:eventOrperC}
Let $x_0\in C$ and let $\left(x_n\right)$ be the corresponding orbit in $[F,G]$.
\begin{itemize}
\item[(a)] If $x_n-x_0\in\mathbb Z$ for some $n\geq 1$, then the orbit of $z_0=e(x_0)$ in $[f,g]$ is periodic.
\item[(b)] If $x_0\in C$ is a periodic point in $[F,G]$, then its period is equal to the  smallest positive integer $N$
such that $$x_N-x_0\in \mathbb Z.$$
\end{itemize}
\end{lemma}

\begin{proof} To prove Part (a), take into account that if $x_{n}\equiv x_0,$ since $x_n\in C$ automatically
Lemma~\ref{L:Candi} ensures that $F(x_n)\equiv G(x_n)$ and therefore $F(x_n)\equiv G(x_n)\equiv F(x_0)\equiv G(x_0)$. Then
$x_{n+1}\equiv x_1$ and, in a similar way, $x_{n+j}\equiv x_j$ for all $j\geq 1.$ From here, it is easily seen that $e(x_0)$ is a periodic point in $[f,g]$, and as a direct consequence we also deduce Part (b).
\end{proof}

For a set $A\subset \R$ and $x\in \R$, we denote $x+A=\{x+a:a\in A\}$.

\begin{lemma}\label{L:Iter_mk}
Let $f,g\in C(\S^1)$ be affine maps with liftings $F(x)=d_1x+\alpha$ and $G(x)=d_2x+\beta$, with
$d_1\neq d_2$ and $\kappa \in \Z$. Let $x_0=\frac{m+\beta-\alpha}{d_1-d_2}\in C$, where $m\in\{1,\ldots,|d_1-d_2|\}.$
If $(x_n)$ is the orbit of $x_0$ generated by the alternating system $[F,G]$, then:
\begin{itemize}
 \item  If $d_2\not=1$
$$x_n \in \frac{d_2^n m + \kappa \frac{1-d_2^n}{1-d_2} + \beta -\alpha}{d_1-d_2} + \mathbb{Z} \,\,\,\textrm{ for all } n\geq 0.$$
\item  If $d_2=1$
$$x_n \in \frac{m + \kappa n + \beta -\alpha}{d_1-1} + \mathbb{Z} \,\,\, \textrm{ for all } n\geq 0.$$
\end{itemize}

\end{lemma}

\begin{proof}
By iterating,
	\begin{eqnarray*}
	x_1&=&F(x_0)=d_1 x_0 +\alpha = \frac{d_1 m + d_1 (\beta-\alpha) + \alpha (d_1-d_2)} {d_1-d_2}\\
	&=&\frac{(d_1 - d_2)m + d_2 m+ \beta (d_1 -1) - \alpha (d_2- 1) +\beta -\alpha} {d_1-d_2}\\
	&=& m + \frac{d_2 m+ \kappa +\beta -\alpha} {d_1-d_2} \in \frac{d_2 m + \kappa + \beta -\alpha}{d_1-d_2} + \mathbb{Z},
	\end{eqnarray*}
	and $x_2=G(x_1)\in G(\frac{d_2 m+ \kappa +\beta -\alpha} {d_1-d_2})+\mathbb{Z}.$
	Since \begin{eqnarray*}
	& &G\left(\frac{d_2 m+ \kappa +\beta -\alpha} {d_1-d_2}\right)=
	\frac{d_2^2 m + d_2 \kappa +d_2 (\beta-\alpha)} {d_1-d_2} +\beta  \\
	&=& \frac{d_2^2 m + d_2 \kappa +d_2 (\beta-\alpha) + \beta(d_1-d_2)} {d_1-d_2}
	=  \frac{d_2^2 m + d_2 \kappa +d_1 \beta - d_2 \alpha} {d_1-d_2} \\
	&=& \frac{d_2^2 m + d_2 \kappa +\beta (d_1 -1)- \alpha (d_2 - 1) +\beta - \alpha} {d_1-d_2}
	=  \frac{d_2^2 m + \kappa (d_2 + 1)+\beta - \alpha} {d_1-d_2},
	\end{eqnarray*}
	we deduce that $x_2\in \frac{d_2^2 m + \kappa (d_2 + 1)+\beta - \alpha} {d_1-d_2} + \mathbb{Z}.$
	
Now, the proof concludes easily by induction.
\end{proof}

Next, we present some properties relative to congruence of integer numbers. Following \cite{Apo}, if $a,b$ are integers,
we take $\gcd(a,b)$ as the non-negative integer $d$ such that $d$ is a common divisor of $a$ and $b$ and
every common divisor divides $d$. Notice that $\gcd(a,b)=0$ if, and only if, $a=b=0$; otherwise, $d\geq 1$. Moreover,
when we write a congruence modulo $q$,$\modu(q)$, it must be understood that $q>0$.

\begin{lemma}\label{L:Congru}
Let $j,q\in\Z\menos\{0\}$, $q>0$. Then  $\mathrm{gcd}(j,q)=1$ if and only if $j^n \equiv 1\modu  (q)$ for some
$n\in\N.$
\end{lemma}
\begin{proof}
Let $d=\gcd(j,q).$ If $d=1$, the well-known Euler-Fermat's theorem establishes that
$j^{\varphi(q)}\equiv 1 \modu (q)$, where $\varphi(\cdot)$ is the Euler function, that is,
$\varphi(n)=\mathrm{Card}\{m\in\N : \gcd(m, n) = 1\}$ (for instance, see \cite{Apo} for a proof; realize that
$\varphi(n)$ is even for all $n\geq 3$, so we can take a general integer $j$ different from $0$).

On the other hand, the condition $j^n \equiv 1\modu  (q)$ for some positive integer $n$ is equivalent to $j^n -uq =1$ for some $u\in\mathbb{Z}$.
As $\frac1d=\frac{j^n}{d} -\frac{uq}{d}\in\Z$, we conclude that $d=1.$
\end{proof}

\begin{lemma}\label{L:Congru2}
Let $q$ be a positive integer. Let $j,m\in\Z\menos\{0\}.$ Then the following statements are equivalent:
\begin{enumerate}
\item $mj^{n}\equiv m\modu (q)$ for some $n\geq 1$,
\item $j^{n}\equiv 1\modu  ( \frac{q}{\mathrm{gcd}(q,m)})$ for some $n\geq 1,$
\item $\gcd\left(j,\frac{q}{\mathrm{gcd}(q,m)}\right)=1.$
\end{enumerate}
\end{lemma}

\begin{proof}
By Lemma~\ref{L:Congru}, we have that (2) and (3) are equivalent. On the other hand, the congruence $mj^{n}\equiv m\modu (q)$ is equivalent to the equation
$m(j^n-1)=uq$ for some $u\in\mathbb{Z}$. Then $\frac{m}{\gcd(q,m)}(j^n-1)=u\frac{q}{\gcd(q,m)}.$ Taking into account that
$\gcd\left(\frac{m}{\gcd(q,m)}, \frac{q}{\gcd(q,m)}\right)=1,$
we obtain  the equivalence between (1) and (2).
\end{proof}

Next lemma will play an important role in the proof of Proposition \ref{P:orders}.
The reader should have in mind the meaning of the integers $d_1$, $d_2$ and $\kappa$.

\begin{lemma}\label{L:mnew}
Let $d_1,d_2\in \Z$, $d_1\neq d_2$, $d_2\not\in\{0,1\}$. Fix $m\in\{1,\ldots,|d_1-d_2|\}$, with $m(d_2-1)+\kk\not=0.$
 Then:
\begin{enumerate}
\item There exists a minimal value $n_0\geq 0$ such that
\begin{equation}\label{Eq:n0}
m d_2^{n_0}+\kappa \frac{d_2^{n_0}-1}{d_2-1}\equiv m d_2^{n}+\kappa \frac{d_2^{n}-1}{d_2-1}\,\,\modu (|d_1-d_2|)
\end{equation}
for some $n> n_0$.
\item Let $n_{m,\kappa}>n_0$ be the smallest value $n>n_0$ for which (\ref{Eq:n0}) holds. Then $\varepsilon_{m,\kappa}:=n_{m,\kappa}-n_0$ is the smallest positive integer $\ell$ such that
\begin{equation}\label{Eq:nmk}
d_2^{\ell}\equiv 1\modu  \left(\frac{|(d_1-d_2)(d_2-1)|}{\gcd\left((d_1-d_2)(d_2-1),d_2^{n_0}(m[d_2-1]+\kappa)\right)}\right).
\end{equation}
\item The integer $n_0\geq 0$ defined in (1) is the minimal non-negative integer $N$ such that
\begin{equation}\label{Eq:n0Calculo}
\gcd\left(d_2,\frac{(d_1-d_2)(d_2-1)}{\gcd\left((d_1-d_2)(d_2-1),d_2^{N}(m[d_2-1]+\kappa)\right)}\right)=1.
\end{equation}
\end{enumerate}
 \end{lemma}
\begin{proof}
(1) As the set $\left\{m d_2^{r}+\kappa \frac{d_2^{r}-1}{d_2-1} \modu(|d_1-d_2|):r\in \N\right\}$ is finite, then $n_0$ and $n$ exist, and then (1) follows.

(2) Similarly, we prove the existence of $n_{m,\kappa}> n_0$ being the smallest value fulfilling that
$$m d_2^{n_0}+\kappa \frac{d_2^{n_0}-1}{d_2-1}\equiv
m d_2^{n_{m,\kappa}}+\kappa \frac{d_2^{n_{m,\kappa}}-1}{d_2-1}\modu (|d_1-d_2|).$$
Equivalently:
\begin{eqnarray*}
 (d_2^{n_{m,\kappa}}-d_2^{n_0})m+\kappa\frac{d_2^{n_{m,\kappa}}-d_2^{n_0}}{d_2-1}
 \equiv 0\modu (|d_1-d_2|)
 \\
 \Leftrightarrow
 d_2^{n_0}(d_2^{n_{m,\kappa}-n_0}-1)m+\kappa d_2^{n_0}\frac{d_2^{n_{m,\kappa}-n_0}-1}{d_2-1}
 \equiv 0\modu (|d_1-d_2|)
 \\
 \Leftrightarrow
(d_2^{n_{m,\kappa}-n_0}-1)
d_2^{n_0}\left[
m+\kappa \frac{1}{d_2-1}
\right]
 \equiv 0\modu (|d_1-d_2|)
 \\
 \Leftrightarrow
d_2^{n_{m,\kappa}-n_0}
d_2^{n_0}\left[
m (d_2-1)+\kappa
\right]
 \equiv
 d_2^{n_0}\left[
m (d_2-1)+\kappa
\right]
\modu \left(|(d_1-d_2)(d_2-1)|\right).
\end{eqnarray*}

Now, by Lemma~\ref{L:Congru2} (notice that $m(d_2-1 ) +\kappa\neq 0$)
$$
d_2^{\varepsilon_{m,\kappa}}=d_2^{n_{m,\kappa}-n_0}
\equiv
1
\modu
\left(
\frac{|(d_1-d_2)(d_2-1)|}{
\gcd\left(
(d_1-d_2)(d_2-1),
d_2^{n_0}\left[
m (d_2-1)+\kappa
\right]
\right)}
\right)
$$
we prove (2), since if it would exist a positive value $\ell$ smaller than $\varepsilon_{m,\kappa}$
holding $(\ref{Eq:nmk})$, then by the above equivalences $n_\ell:=n_0+\ell<n_{m,\kappa}$ would satisfy $(\ref{Eq:n0})$, in contradiction
with the minimality of $n_{m,\kappa}.$

(3) Finally,  Lemma~\ref{L:Congru} implies that
$$\gcd\left(d_2,
\frac{(d_1-d_2)(d_2-1)}{
\gcd\left(
(d_1-d_2)(d_2-1),
d_2^{n_0}\left[
m (d_2-1)+\kappa
\right]
\right)}
\right)=1,$$ which concludes the proof.
\end{proof}

\begin{definition}[$\varepsilon_{m,\kk}$] Given  $d_1,d_2\in \Z$, $d_1\neq d_2$ with $\kk=\kkdef\in\Z$,
let $m\in\{1,2,\ldots,|d_1-d_2|\}.$ Then:
\label{D:defepsilonmk}
\begin{equation}\label{E:defepsilommk}
 \varepsilon_{m,\kk}:=\left\{
 \begin{array}{lcl}
 1&&\textrm{ if } m(d_2-1)+\kk=0,\\
 1&&\textrm{ if } m(d_2-1)+\kk\not=0,d_2=0,\\
N
&&\textrm{ if } m(d_2-1)+\kk\not=0,d_2=1,\\
  \textrm{the value of Lemma~\ref{L:mnew}-(2)}&&\textrm{otherwise.}
 \end{array}
 \right.
\end{equation}

The integer $N>0$ is the smallest positive integer $n$ satisfying
$$n\kappa\equiv 0\modu\left(|d_1-1|\right).$$
\end{definition}

Next, we give a characterization of periodic orbits shared by two affine circle maps with different degrees.

\begin{proposition}\label{P:orders}
Let $f,g\in C(\S^1)$ be affine maps with liftings $F(x)=d_1 x + \alpha$ and $G(x)=d_2 x +\beta,$ with $d_1\neq d_2$
and such that  $\kappa \in\Z$.																
Let $x_0=\frac{m + \beta -\alpha}{d_1-d_2}\in C$, with $m\in\{1,\ldots,|d_1-d_2|\}.$ Then:
\begin{itemize}
\item[(a)] $e(x_0)$ is an eventually periodic point shared by $f$ and $g$.
\item[(b)] If $\ell\in\N$ is the smallest non-negative integer holding that
$f^\ell(e(x_0))=g^\ell(e(x_0))$ is periodic, then its period is
$\lambda=\varepsilon_{m,\kk}$  defined by  Equation~(\ref{E:defepsilommk}).
\end{itemize}
\end{proposition}
\begin{proof}
Let  $(x_n)$ be the orbit of $x_0\in C$ generated by $[F,G]$,
then by Lemma~\ref{L:RCandi} the orbit $(e(x_n))$ in $\mathbb{S}^1$, generated by $\sisalter$, is either periodic or eventually periodic.
Notice that this (eventually) periodic orbit is shared by $f$ and $g$ because,
by Lemma~\ref{L:Candi}, $F^n(x_0)-G^n(x_0)\in\Z$ for all $n\in \N$.
Let $n_0\geq 0$ be the first integer such that $f^{n_0}(e(x_0))$ is periodic.

Now we distinguish these cases.
\begin{description}
\item[(1)] \underline{$m(d_2-1)+\kk=0.$} We have 
\begin{eqnarray*}
 x_1&=&
F(x_0)=d_1\left(\frac{m + \beta -\alpha}{d_1-d_2}\right)+\alpha=\frac{md_1+\beta d_1 -\alpha d_2}{d_1 - d_2}\\
&=&\frac{md_1+\kappa +\beta-\alpha}{d_1-d_2}=\frac{md_1 +m(1-d_2)+\beta-\alpha}{d_1-d_2}\\
&=&m+\frac{m + \beta -\alpha}{d_1-d_2} \in x_0 +\Z,
\end{eqnarray*}
 and then $x_1-x_0\in\Z$. Now, Lemma~\ref{L:eventOrperC} implies that the sequence $(x_n)$ generated by $[F,G]$
is a lifted cycle and $z=e(x_0)$ is a cycle of order $\varepsilon_{m,\kk}=1$ of $\sisalter$.

\item[(2)] \underline{$ m(d_2-1)+\kk\not=0 , d_2=0.$}  In this case, for a given $x_0\in C$ the lifted orbit is
$$x_0,x_1=d_1x_0+\alpha,x_2=\beta,x_3=d_1\beta+\alpha,x_4=\beta,x_5=d_1\beta+\alpha,\dots$$
Observe that $\kk=\beta(d_1-1)+\alpha,$ so $d_1\beta +\alpha=\kk +\beta,$ and consequently
$x_n-x_{n-1}=(-1)^n\kk\in\Z$ for any $n\geq 3$. Therefore, the orbit of $e(f^2(x_0))$ is periodic
of period $\varepsilon_{m,\kk}=1$.

\item[(3)] \underline{$ m(d_2-1)+\kk\not=0 , d_2=1.$}
Let $x_0\in C$ and let $(x_n)$ be the orbit generated by $[F,G]$. According to Lemma~\ref{L:RCandi}, let $n_0$ be the smallest
non-negative integer such that $f^{n_0}(e(x_0))$ is a periodic point in $[f,g]$ or, equivalently, such that $x_{n_0}$ generates
a periodic lifted cycle in $[F,G]$. By Lemma~\ref{L:Iter_mk},
 $$x_n\in  \frac{m + \kappa n + \beta -\alpha}{d_1-1} + \mathbb{Z}=x_0 +\frac{\kappa n }{d_1-1}+\mathbb Z$$
 and $x_n-x_{n_0}\in\Z$ implies
$\frac{(n-n_0)\kk}{d_1-1}\in\Z$. By Lemma~\ref{L:eventOrperC} and the last congruence, the period is equal to
the first positive integer $N$ satisfying $N\kk\equiv 0 \modu(|d_1-1|)$, the value of $\varepsilon_{m,\kk}$
in Definition~\ref{D:defepsilonmk}.

\item[(4)]  \underline{$ m(d_2-1)+\kk\not=0 , d_2\notin\{0,1\}.$}
Then, by Lemma~\ref{L:mnew} (cf. Equation~(\ref{E:defepsilommk}) in Definition~\ref{D:defepsilonmk}),  the period of $f^{n_0}(e(x_0))$
is given by $\varepsilon_{m,\kk}$, as described in the statement.

\end{description}

\end{proof}

Thus, we have proved the following result:

\begin{theorem}\label{T:oddLambda}
Let $f,g\in C(\S^1)$ be affine circle maps with liftings $F(x)=d_1 x + \alpha$ and $G(x)=d_2 x +\beta,$ with $d_1\neq d_2$ and such that
  $\kappa=\beta(d_1-1)-\alpha(d_2-1)\in\Z$.	
Then the set $\Lambda$ of odd periods from $\mathrm{Per}[f,g]$ is given by
$$\Lambda=\{\lambda=\varepsilon_{m,\kk}: 1\leq m\leq |d_1-d_2|, \lambda \text{ odd }\},$$
where $\varepsilon_{m,\kk}$ is  introduced in Definition~\ref{D:defepsilonmk}.
\end{theorem}

The case $\Lambda=\emptyset$ can occur for degree $-1$, for instance consider the affine liftings $F(x)=-x+\alpha, G(x)=x+\beta$, 
with $\beta=\frac{q}{2}$, $q$ odd; here, $d_1=-1, d_2=1$, $\kappa=-2\beta\in\Z$ and $C$ is reduced to two points, 
$\frac{1+\beta-\alpha}{-2}$ and $\frac{2+\beta-\alpha}{-2}$; it is easily seen that
both points generate a periodic sequence of order $2$, an even number.

Nevertheless,  under the conditions of Theorem~\ref{T:oddLambda}: %
\begin{corollary}
Let $d_1=-1$ and assume that $\Lambda\neq\emptyset$. Then $\Lambda=\{1\}.$
\end{corollary}

\begin{proof}
Let $m\in\{1,2,\ldots,|d_1-d_2|\}.$ If $d_1=-1$, then $\kappa=-2\beta - \alpha (d_2-1)\in\mathbb Z$. Fix $m$.
Now, we apply Definition~\ref{E:defepsilommk} and Theorem~\ref{T:oddLambda}. We distinguish, according to the referred definition,
the following cases:

-- If $m(d_2-1)+\kappa=0$ or  $m(d_2-1)+\kappa\neq 0, d_2=0$, then $\varepsilon_{m,\kappa}=1$.

-- If $m(d_2-1)+\kappa\neq 0, d_2=1$, then $\kappa=-2\beta\in\mathbb{Z}$, and we know that, from Definition~\ref{E:defepsilommk},
we have to add the period $N$, the smallest positive integer such that $n\kappa\equiv 0 \modu(|d_1-1|)$, that is, $-2n\beta\equiv 0\modu(|d_1-1|).$
Thus, either $N=1$ or $N=2$. If $N=1$, we add it to $\Lambda$. If $N=2$, this case does not originate a value of $\Lambda$.

-- If $m(d_2-1)+\kappa\neq 0, d_2\notin\{0,1\}, d_2\neq d_1$, by Theorem~\ref{T:oddLambda} and its previous Lemma~\ref{L:mnew},
firstly we have to compute the smallest non-negative integer value $n_0$ such that
$$\gcd\left(d_2,\frac{(d_1-d_2)(d_2-1)}{\gcd\left((d_1-d_2)(d_2-1),d_2^{N}(m[d_2-1]+\kappa)\right)}\right)=1.$$
We short the notation by setting $d=d_2$. Since $d_1=-1$, the above condition is written as
\begin{equation}\label{Eq:cn0}
\gcd\left(d,\frac{1-d^2 }{\gcd\left( 1-d^2, d^{N}(m[d -1]+\kappa)\right)}\right)=1.
\end{equation}
Notice that $1=\gcd(d,d+1)=\gcd(d,d-1)$ if $|d|\geq 2$. From here,
if we put $w:=\frac{1-d^2 }{\gcd\left( 1-d^2, d^{N}(m[d -1]+\kappa)\right)}$, obviously $\gcd(d,w)=1$ since each prime factor of $w$
is a divisor of $1-d^2$ and we know that $\gcd(d,1-d^2)=1.$ Therefore, (\ref{Eq:cn0}) is satisfied by all the non-negative integers and
consequently $n_0=0$. Now, in order to evaluate $\varepsilon_{m,\kappa}=n_{m,\kappa}-n_0=n_{m,\kappa},$ we have to obtain the smallest positive
integer $N$ holding $d^N\equiv 1\modu(\frac{|1-d^2|}{\gcd\left(1-d^2, m(d-1)+\kappa\right)}).$ Taking into account
that $d^2 -1 = \gcd\left(1-d^2, m(d-1)+\kappa\right)\frac{d^2-1}{\gcd\left(1-d^2, m(d-1)+\kappa\right)}$, we find that either $N=1$ or $N=2$.
If $N=2$, this case does not provide a value of $\Lambda$. If $N=1$, we add it to $\Lambda$.

Summarizing, with the study of the above cases, we have seen that the unique odd period we can obtain in $\Lambda$ is precisely $1$.
\end{proof}

\begin{center}
\begin{table}[ht]
\begin{tabular}{|c|c|c|c|c|c|c|
c|}
\hline
$d_1$&$d_2$&$\alpha$&$\beta$&$|d_1-d_2|$& $\kappa$&$\Lambda$
\\ \hline
38&7&0&0&31& 0&$\{1,15\}$
\\ \hline
46&16&0&1/45&30& 1&$\{15\}$
\\ \hline
16&6&0&1/15&10& 1&$\{5\}$
\\ \hline
39&16&0&0&23& 0&$\{1,11\}$
\\ \hline
5&36&0&0&31& 0&$\{1,3\}$
\\ \hline
10&4&0&1/9&6& 1&$\{3\}$
\\ \hline
31&256&2&7&225& -300&$\{1,3,5,15\}$
\\ \hline
10&2&0&1/9&8& 1&$\{1\}$
\\
\hline
31&1&23&1/5&30&6&$\{5\}$
\\ \hline
-1&9&1/4&1&10&-4&$\{1\}$
\\ \hline
\end{tabular}
\caption{Examples of odd periods of alternating systems after applying the proposed algorithm.}
\label{Table:algoritmo}
\end{table}
\end{center}

\section{The algorithm}\label{S:algorithm}

Although Theorem~\ref{T:oddLambda}  provides the characterization of odd periods for an alternating system, often it
is necessary to do a certain number of hard computations to obtain it. Hence, we are going to give
an algorithm to do these calculations with a computer. The program introduced in next subsection has been implemented in \emph{Maxima}, which is a free software of symbolic calculation.

The last function of the program, ComputeLambda, must be called with four arguments ($d_1$, $d_2$, $\alpha$ and  $\beta$)
and returns the set $\Lambda$ of odd periods. This function calls another one called Epsilonmkappa which distinguishes in what case of Definition~\ref{D:defepsilonmk} we are and computes the value of $\varepsilon_{m,\kappa}$. This value is easy to find  when we are in the two first cases of the definition of $\varepsilon_{m,\kappa}$, however  when we are in the third and fourth cases   stronger calculations are needed and made respectively by  the functions Epsilon3 and Epsilon4.

Epsilon3 is an easy function that uses a loop in order to find the least integer $N$ satisfying $N\kappa\equiv 0\modu\left(|d_1-1|\right)$. In exchange, in order to call Epsilon4 from the function Epsilonmkappa we need previously to know the value of $n_0$ satisfying the conditions of Lemma~\ref{L:mnew}-(i). The function Computen0 searches this number recursively, next we call the function Epsilon4 with  $n_0$ and, by means of a loop for, we find the smallest integer $n>n_0$ satisfying the condition in Lemma~\ref{L:mnew}-(ii). This completes the computation of $\varepsilon_{m,\kk}$.

ComputeLambda  works by means of a loop for between the numbers 0 and $|d_1-d_2|$. 
 It checks if the  calculated value by Epsilonmkappa is odd and in this case it adds the number to the set $\Lambda$.

\newcommand{\comentariomaxima}[1]{\textrm{$/*$ #1 $*/$}}
\newcommand{\fincomandomaxima}{\$}
\newcommand{\espaciovertical}{\vspace*{.25cm}}

\subsection*{The algorithm implemented in Maxima.\newline}
\hfill

\noindent\comentariomaxima{In the following definition we compute, by means of a recursive function, the value $n_0$ given by Lemma~\ref{L:mnew}--(1) and (3). This
value is needed to obtain $\varepsilon_{m,\kappa}$ in the fourth alternative of Definition~\ref{D:defepsilonmk}
which we implement in function Epsilon4 below. We call the function with $n=0$, it  checks if the greatest common divisor in Proposition~\ref{P:orders} is
1 and in this case $n_0=0$; if not, the function calls itself with $n=1$, do the test for $n=1$ and so on}

\espaciovertical

\noindent\texttt{
\noindent{}Computen0(d1,d2,m,k,n):=
\newline\indent if gcd(d2,(d1-d2)*(d2-1)/gcd((d1-d2)*(d2-1),d2\^{}n*(m*(d2-1)+k)))=1
\newline\indent \indent then n
\newline\indent \indent else
Computen0(d1,d2,m,k,n+1)\fincomandomaxima
\espaciovertical
}

\noindent\comentariomaxima{The following function, Epsilon4, computes the value of $\varepsilon_{m,\kappa}$ when the values of $d_1,d_2,\alpha$ and $\beta$ are in the condition of the fourth alternative in Definition~\ref{D:defepsilonmk}. This function will be called by Epsilonmkappa}

\espaciovertical
\noindent
\texttt{Epsilon4(d1,d2,k,n0,m):=block(
\newline\indent l:1,
\comentariomaxima{l is the variable used to obtain the value of  $\ell$}
\newline\indent
for n:1 thru abs(d1-d2) while abs(remainder(
\newline\indent\indent d2\^{}n-1,(d1-d2)*(d2-1)/gcd((d1-d2)*(d2-1),d2\^{}n0*(m*(d2-1)+k))))$>$0
\newline\indent\indent do
l:l+1,
\newline\indent
l \comentariomaxima{Epsilon4 returns the value of $\ell$}
\newline)\fincomandomaxima\comentariomaxima{End of definition of Epsilon4}
}
\espaciovertical

\noindent\comentariomaxima{Function Epsilon3, defined below, computes the value of $\varepsilon_{m,\kappa}$ when the values of $d_1,d_2,\alpha$ and $\beta$ are in third alternative in Definition~\ref{D:defepsilonmk}. This function will be called by Epsilonmkappa}

\espaciovertical
\noindent
\texttt{Epsilon3(d1,k):=block(
\newline\indent
N:1, \comentariomaxima{N is the variable used to obtain the value of  $N$}
\newline\indent
modcong:abs(d1-1),
\newline\indent
    for n:1 thru modcong while
    \newline\indent \indent
    abs(remainder(n*k,modcong))$>$0
    \newline\indent \indent
    do N:N+1,
    \newline\indent
    N  \comentariomaxima{Epsilon3 gives back the value of $N$}
    \newline
    )\fincomandomaxima \comentariomaxima{End of definition of Epsilon3}
}
\espaciovertical

\noindent\comentariomaxima{Function Epsilonmkappa computes the value of $\varepsilon_{m,\kappa}$ distinguishing the
four possibilities given in Definition~\ref{D:defepsilonmk}}


\espaciovertical
\noindent
\texttt{
Epsilonmkappa(d1,d2,alpha,beta,m):=block(
\newline\indent     emk:0,
\newline\indent         k:beta*(d1-1)-alpha*(d2-1), \comentariomaxima{the value of $\kappa$}
    \newline\indent     condition:m*(d2-1)+k, \comentariomaxima{the value needed to distinguish the cases of Definition~\ref{D:defepsilonmk}}
    \newline\indent     if condition=0 then emk:1 else(
        \newline\indent\indent     if d2=0 then emk:1 else(
         \newline\indent\indent\indent            if d2=1 then emk:Epsilon3(d1,k) else(
                    \newline\indent\indent\indent\indent n0:Computen0(d1,d2,m,k,0),
                    \newline\indent\indent\indent\indent  emk:Epsilon4(d1,d2,k,n0,m)
                \newline\indent\indent\indent)
  \newline\indent\indent          )
   \newline\indent ),
   \newline\indent
emk \comentariomaxima{we return the value of $\varepsilon_{m,\kappa}$}
\newline)\fincomandomaxima \comentariomaxima{End of definition of Epsilonmkappa}
}
\espaciovertical

\noindent\comentariomaxima{Function ComputeLambda calculates the set $\Lambda$
}

%
%


\espaciovertical
\noindent
\texttt{
ComputeLambda(d1,d2,alpha,beta):=block(
\newline\indent difference:abs(d1-d2),
\newline\indent k:beta*(d1-1)-alpha*(d2-1),
\newline\indent Lambda:[], \comentariomaxima{the vector that will contain the odd period}
\newline\indent if integerp(k) and not(d1=d2) then ( \comentariomaxima{ we check if k is integer}
\newline\indent \indent    for m:0 thru difference do(
\newline\indent \indent\indent        emk:Epsilonmkappa(d1,d2,alpha,beta,m),
\newline\indent \indent\indent        candidate:emk,
\newline\indent \indent\indent        if oddp(candidate) then Lambda:append(Lambda,[candidate])
\newline\indent \indent\indent \comentariomaxima{we only add candidate to Lambda when it is odd}
\newline\indent \indent    )
\newline\indent ) else Lambda:[],
\newline\indent setify(Lambda) \comentariomaxima{the function returns the set made of the components of the components from Lambda}
\newline )\fincomandomaxima \comentariomaxima{End of definition of the function ComputeLambda}
}
\espaciovertical

\noindent\comentariomaxima{  Finally, the computation of the odd periods of an alternating affine system with values $d_1,d_2,\alpha,\beta$ will be obtained by executing the order:}

\espaciovertical
\noindent
\texttt{
ComputeLambda(d1,d2,alpha,beta)
\fincomandomaxima
}

\espaciovertical

\newcommand{\datos}[8]{#1&#2&#3&#4&#5&#6&#7&$#8$
\hline
\\}

We apply this algorithm to some alternating systems and we obtain some additional examples presented in Table~\ref{Table:algoritmo}.

\section{Conclusions}
The algorithm given in this manuscript is a valuable instrument for computing the odd periods in
$\mathrm{Per}(\sisalter)$ when considering affine circle maps, $f$ and $g$,  and it finishes the  problem of characterizing $\mathrm{Per}(\sisalter)$. However we think that in  this line of research a lot of work can be done in the future. We mention some valuable problems to analyze.

Admittedly the problem of characterizing  $\mathrm{Per}(\sisalter)$  for continuous circle maps $f,g$ seems too much ambitious. 
We propose to make a deep analysis of alternating systems $\sisalter$ for circle homeomorphisms since the sets of periods for these maps are simple and depend on the rotation number of the homeomorphisms when the degree is 1. However,  the triviality of the sets of  periods of homeomorphisms does not guarantee an easy control of the set  $\Per(\sisalter)$ since the periods of $f\circ g$ strongly  depends on the rationality of its rotation number. In this point it is interesting to remark that the composition of two   degree 1 homeomorphisms with irrational numbers  (and then with empty set of periods) can give a degree 1 
homeomorphism with rational number (and then with nonempty set of periods), see \cite[Chapter 1, Section 4]{demelovanstrien}.

In the setting of affine circle maps we still propose to make a deep research by considering alternating systems of more than two maps.

\def\cprime{$'$}

\section*{Acknowledgments}
Authors have been partially supported by the Grants MTM2014-52920-P and MTM2017-84079-P from Agencia Estatal de Investigaci\'on (AEI) y Fondo Europeo de Desarrollo Regional (FEDER). We acknowledge the referees for their suggestions which allowed us to improve the reading of the paper.

\end{document}